\documentclass[12pt]{article}
\usepackage{amsmath}
\usepackage{amsfonts}
\usepackage{amscd}
\usepackage{amsthm}
\usepackage{amssymb}
\usepackage{latexsym}
\usepackage{graphicx}
\usepackage{graphics}
\newtheorem{theorem}{Тheorem}[section]
\newtheorem{lemma}{Lemma}

\newcommand{\nc}{\newcommand*}

\usepackage[cp1251]{inputenc}
\inputencoding{cp1251}
\usepackage[T2B]{fontenc}
\usepackage[english]{babel}
\nc{\beq}{\begin{equation}}
\nc{\eeq}{\end{equation}}
\nc{\lb}[1]{\label{#1}}
\nc{\nn}{\nonumber}
\nc{\ts}{\textstyle}
\nc{\ds}{\displaystyle}
\nc{\lar}{\leftarrow}
\nc{\rar}{\rightarrow}
\nc{\Rar}{\Rightarrow}
\nc{\half}{\frac{1}{2}}
\nc{\reff}[1]{(\ref{#1})} 
\nc{\basc}[2]{\left[#1_{#2}\right]_{#2=0}^{\infty}}
\nc{\basb}[2]{\left\{#1_{#2}\right\}_{#2=0}^{\infty}}
\nc{\basbx}[3]{\left\{#1_{#2}(#3)\right\}_{#2=0}^{\infty}}
\nc{\basby}[4]{\left\{#1_{#2}(#4)\right\}_{#3=0}^{\infty}}
\def\wt{\widetilde}
\begin{document}

{\bf Borzov V.V.,  Damaskinsky E.V.}
\vspace{0.5cm}

\centerline{\Large\bf The discrete spectrum of Jacobi matrix}
\vspace{0.5cm}

\centerline{\Large\bf related to recurrence relations}
\vspace{0.5cm}

\centerline{\Large\bf with periodic coefficients\footnote{
The work is partially supported by RFBR grant № 15-01-03148-а}}

\begin{flushright}
{\it Dedicated to Petr Kulish\qquad\qquad\qquad\qquad} \\
{\it in connection with the seventieth anniversary}.
\end{flushright}

\section{Introduction}\lb{s1}

The interest in the study of orthogonal polynomials defined by recurrence relations with periodic 
and asymptotically periodic coefficients increased after the appearance of the article \cite{1} 
about the properties of the periodic Toda lattices (see, for example, \cite{2}-\cite{12}).
In particular, in work \cite{13} recurrence relation with periodic coefficients are investigated.
It was shown that such polynomials can be described by the classical Chebyshev polynomials.
The aim of this work is to study the discrete spectrum of the Jacobi matrix,
connected with polynomials in this class, i.e. polynomials with periodic coefficients
in recurrent relations. As an example, we consider:

 a) the case when period $N$  of coefficients in recurrence relations equals
to three (as a particular case  we consider "parametric" Chebyshev polynomials \cite{17});

 b) the elementary   $N$-symmetrical Chebyshev polynomials
 ($N=3,4,5$), that was introduced  by authors in studying  the "composite model
 of generalized  oscillator" \cite{14}.

 Let us remind some necessary results from \cite{13}. We denote by
 $\basbx{\varphi}{n}{x}$ the polynomial sequence defined by recurrence relations
\beq\lb{de01}
\varphi_{n}(x)=(x+a_{n-1})\varphi_{n-1}(x)-b_{n-1}\varphi_{n-2}(x),
\quad n\geq 1,\quad\varphi_{0}(x)=1,\, \varphi_{-1}(x)=0,
\eeq
where the coefficients are periodic with period $N\geq2$:
\beq\lb{de02}
a_{n+N}=a_n,\qquad b_{n+N}=b_{n},\quad n\geq0.
\eeq
We will use the Chebyshev polynomials of the second kind defined  by recurrence relations
\beq\lb{de03}
tU_{n}(t)=U_{n+1}(t)+U_{n-1}(t),\quad n\geq 0,\quad U_{0}(t)=1,\, U_{-1}(t)=0.
\eeq
It was proved in \cite{13} that for any  $N\geq1$ the polynomial $\varphi_{N-1}(x)$ divides the polynomial  $\varphi_{2N-1}(x)$
\beq\lb{de04}
\varphi_{2N-1}(x)=\varphi_{N-1}(t)P_{N}(x),
\eeq
where  the polynomial $P_{N}(x)$ defined from equality \reff{de04}.
Besides, recurrence relations
\beq\lb{de06}
\varphi_n(x)=\varphi_{Nm+k}(x)=\varphi_{k+N}(x)U_{m-1}(P_N(x))-\varphi_k(x)U_{m-2}(P_N(x)).
\eeq
are fulfilled for $n=Nm+k,\, k=\overline{0;(N-1)},\, m\geq2$.
The Jacobi matrix $J=\basc{j^{(N)}}{i,k}$ related to recurrence relations    \reff{de01} one can write in following form
\beq\lb{de07}
j_{i,k}^{(N)}=A \delta_{i+1,k}+B \delta_{i,k}+C \delta_{i-1,k},
\eeq
where the matrixes
\small{
\begin{gather}
A=\begin{pmatrix}
0&0&\ldots&0\\
0&0&\ldots&0\\
\ldots&\ldots&\ldots&\ldots\\
1&0&\ldots&0
\end{pmatrix},\qquad
C=\begin{pmatrix}
0&0&\ldots&0&b_{N-1}\\
\ldots&\ldots&\ldots&\ldots&\ldots\\
0&0&\ldots&0&0\\
0&0&\ldots&0&0
\end{pmatrix}, \nn \\[.5cm]
B=\begin{pmatrix}
-a_0&1&0&\ldots&0&0\\
b_0&-a_1&1&\ldots&0&0\\
\hdotsfor{6}\\
0&0&0&\ldots&0&0\\
0&0&0&\ldots&-a_{N-2}&1\\
0&0&0&\ldots&b_{N-2}&-a_{N-1}
\end{pmatrix},\lb{de08}
\end{gather}
}have the size $N\times N.$ Let $X_{\mu}=(x_1,x_2,\dots,x_n,\ldots)^t\in \ell^2$ be an eigenvector
for matrix $J,$ corresponding to eigenvalue $\mu$:
\beq\lb{de08a}
(J-\mu I) X_{\mu}=0.
\eeq
The following necessary and sufficient condition is hold: a solution $\mu$
of the equation \reff{de08a}  is an eigenvalue of matrix  $J$  if and only if
when
\beq\lb{de09}
\|X_{\mu}\|_2^2=\sum_{n=1}^{\infty}|x_n|^2=\sum_{n=0}^{\infty}|\varphi_n(\mu)|^2<\infty
\eeq
In the next paragraph we will obtain the "critical equation" for Jacobi matrix
 $J$ . A solution of the equation \reff{de09} is called the "critical value" of
Jacobi matrix  $J$ . The necessary condition for $\mu$ to be an eigenvalue of matrix $J$ is that
$\mu$ be a "critical value" of  $J$ .

\section{The critical equation for Jacobi matrix $J$}\lb{s2}

Let us introduce the notation which was needed in the following:
\begin{subequations} \lb{de10}
\begin{eqnarray}
\sigma_{k,k+m}=\sum_{n=k}^{k+m}\varphi_{n}^2(\mu),\quad
S_1=\sigma_{N,2N-1},\quad S_2=\sigma_{0,N-1},\quad  S=S_1+S_2; \lb{de10-1} \\
S_n^N(\mu)=\sigma_{n,n+2N-1},\quad D_n^N(\mu)=\sum_{k=n}^{n+N-1}\varphi_k(\mu)\varphi_{k+N}(\mu)  \lb{de10-2} \\
\Delta_n^N(\mu)=S_n^N(\mu)-P_N(\mu)D_n^N(\mu),\quad n\geq0  \lb{de10-3}
\end{eqnarray} \end{subequations}
It is clear that $S=S_0^N(\mu).$
The following statements are hold:
\begin{lemma}\lb{l1}
Let the polynomial system  $\basbx{\varphi}{n}{x}$ is defined by recurrence relations \reff{de01} and
periodic conditions \reff{de02}. Then for any
 $n\geq 2N$ the following recurrence relations
\beq\lb{de11}
\varphi_n(x)=P_N(x)\varphi_{n-N}(x)-\varphi_{n-2N}(x).
\eeq
are fulfilled.
\end{lemma}
\begin{proof}[Proof]
From relations \reff{de03} and \reff{de06} we have ($n=Nm+k$)
\begin{multline*}
P_N(x)\varphi_{n-N}(x)-\varphi_{n-2N}(x)=\varphi_{k+N}(x)P_N(x)U_{m-2}(P_N(x))\\
-\varphi_{k}(x)P_N(x)U_{m-3}(P_N(x))-\varphi_{k+N}(x)U_{m-3}(P_N(x))+
\varphi_{k}(x)U_{m-4}(P_N(x))=\\
=\varphi_{k+N}(x)\left(U_{m-1}(P_N(x))+\rule[-5pt]{0pt}{20pt}U_{m-3}(P_N(x))\right)-
\varphi_{k}(x)\left(U_{m-2}(P_N(x))\rule[-5pt]{0pt}{20pt}+U_{m-4}(P_N(x))\right)\\
-\varphi_{k+N}(x)U_{m-3}(P_N(x))+\varphi_{k}(x)U_{m-4}(P_N(x))=\\
=\varphi_{k+N}(x)U_{m-1}(P_N(x))-\varphi_{k}(x)U_{m-2}(P_N(x))=\varphi_{Nm+k}(x)=\varphi_{n}(x).
\end{multline*}
\end{proof}
\begin{lemma}\lb{l2}
Let the polynomial system  $\basbx{\varphi}{n}{x}$ is defined by recurrence relations \reff{de01} and
periodic conditions \reff{de02}. Besides, let
 $\Delta_n=\Delta_n^N(\mu)$ is defined by the equalities
 \reff{de10-2} and \reff{de10-3}. Then for all $n\geq0$ the identity
\beq\lb{de12}
\Delta_n=\Delta_0
\eeq
is fulfilled.
\end{lemma}
\begin{proof}[Proof]
For proof this proposition by induction it is sufficiently to show that for all $n\geq0$ the equality
\beq\lb{de13}
\Delta_n=\Delta_{n+1}.
\eeq
is fulfilled. Taking into account \reff{de10-2} and \reff{de10-3}, for proof this relation it is sufficiently
to check the equality
\begin{multline*}
\varphi_{n}^2+\varphi_{n+1}^2+\ldots+\varphi_{n+2N-1}^2-P_N\left(\varphi_{n}\varphi_{n+N}+\varphi_{n+1}\varphi_{n+N+1}+
\ldots\varphi_{n+N-1}\varphi_{n+2N-1}\right)=\\
\varphi_{n+1}^2+\ldots+\varphi_{n+2N-1}^2+\varphi_{n+2N}^2-P_N\left(\varphi_{n+1}\varphi_{n+N+1}+\ldots+
\varphi_{n+N-1}\varphi_{n+2N-1}+\varphi_{n+N}\varphi_{n+2N}\right),
\end{multline*}
This equality is equivalent to the following relation
\beq\lb{de14}
\varphi_{n}^2-P_N\varphi_{n}\varphi_{n+N}=\varphi_{n+2N}^2-P_N\varphi_{n+N}\varphi_{n+2N}.
\eeq
From \reff{de11} it is follow that
\beq\lb{de15}
\varphi_{n+2N}^2=P_N^2\varphi_{n+N}^2-2P_N\varphi_{n}\varphi_{n+N}+\varphi_{n}^2.
\eeq
Substituting \reff{de15} in \reff{de14}, we see that the equality  \reff{de14} is fulfilled.
Therefore the equality  \reff{de13} is also hold.
\end{proof}

\begin{theorem}
For $\mu$ to be an eigenvalue of matrix $J$ defined by \reff{de07} and \reff{de08} it is  necessary that
$\mu$ be a solution of equation
\beq\lb{de16}
\Delta_0^N(\mu)=0.
\eeq
\end{theorem}
\begin{proof}[Proof]
Let us denote
\beq\lb{de17}
\sigma_{k,k+m}=\sum_{n=k}^{k+m}\varphi_{n}^2(\mu).
\eeq
It follows from \reff{de10}, \reff{de11} that
\begin{gather*}
\sigma_{0,2N-1}=S,\\
\sigma_{2N,3N-1}=\left(P_N\varphi_{N}-\varphi_{0}\right)^2+\ldots+\left(P_N\varphi_{2N-1}\varphi_{N-1}\right)^2=
P_N^{\,\,2}S_1+S_2-2P_ND_0^N(\mu).
\end{gather*}
From  \reff{de10-3} it follows that
\begin{equation*}
   S-P_ND_0^N=\Delta_0^N\quad\Rightarrow\quad P_ND_0^N=S-\Delta_0^N,
\end{equation*}
so we have
\begin{equation}\label{de18-1}
\sigma_{2N,3N-1}=2\Delta_0^N+\left(P_N^{\,\,2}-2\right)S_1-S_2.
\end{equation}
Analogously, using lemma 2, we obtain
\begin{equation}\label{de18-2}
\sigma_{3N,4N-1}=2\Delta_0^N+\left(P_N^{\,\,2}-2\right)\sigma_{2N,3N-1}-\sigma_{N,2N-1}.
\end{equation}
Than we have
\begin{equation}\label{de19}
\sigma_{kN,(k+1)N-1}=2\Delta_0^N+\left(P_N^{\,\,2}-2\right)\sigma_{(k-1)N,kN-1}-\sigma_{(k-2)N,(k-1)N-1}.
\end{equation}
Summing these equalities in $k$, we have
\begin{multline*}
\sigma_{0,nN-1}=
\sum_{k=0}^{nN-1}\varphi_{k}^{\,\,2}
=S+2\Delta_0^N+\left(P_N^{\,\,2}-2\right)S_1-S_2\\
+\sum_{k=3}^{n}\left(2\Delta_0^N+\left(P_N^{\,\,2}-2\right)\sigma_{(k-1)N,kN-1}-\sigma_{(k-2)N,(k-1)N-1}\right)
=\\
=2(n-1)\Delta_0^N+\left(P_N^{\,\,2}-1\right)S_1
+\left(P_N^{\,\,2}-2\right)\left(\sigma_{0,nN-1}-S\right)
-\sigma_{0,nN-1}+S_2+\sigma_{(n-1)N,nN-1}.
\end{multline*}
From here we obtain the relation
\begin{equation}\label{de20}
(4-P_N^{\,\,2})\sum_{k=0}^{nN-1}\varphi_{k}^{\,\,2}=2(n-1)\Delta_0^N+S_1+(3-P_N^{\,\,2})S_2+\sigma_{(n-1)N,nN-1}.
\end{equation}
for finding the quantity
$\sigma_{0,nN-1}=\sum_{k=0}^{nN-1}\varphi_{k}^{\,\,2}$ .
It is clear that if $\sum_{k=0}^{nN-1}\varphi_{k}^{\,\,2}(\mu)<\infty,$ then
\begin{equation}\label{de21}
\lim_{n\rightarrow\infty}\sigma_{(n-1)N,nN-1}=0.
\end{equation}
Then it follows from \reff{de20} that the series
 $\sum_{k=0}^{nN-1}\varphi_{k}^{\,\,2}(\mu)$ is convergent if
\begin{equation*}
\Delta_0^N(\mu)=0.
\end{equation*}
\end{proof}
\bigskip

{\bf Remark 1.} All eigenvalues of the matrix $J$ must satisfy the "critical equation" \, \reff{de16}.

But it is possible that some critical values of matrix $J$ are not satisfied the necessary and sufficient
condition \reff{de09}, i.e.
the corresponding vector $X_{\mu}$ are not belonging to $\ell_2.$
\bigskip

{\bf Remark 2.} Apparently for any $N\geq2$
 the polynomial  $\varphi_{N-1}(\mu)$  divides the  $\Delta_0^N(\mu)$, i.e.the equality
\begin{equation}\label{de22}
\Delta_0^N(\mu)=\varphi_{N-1}(\mu)\,Q_N(\mu),
\end{equation}
is true.
Then the "critical equation"  \reff{de16} splits into two equations
\begin{equation}\label{de23}
\varphi_{N-1}(\mu)=0,
\end{equation}
and
\begin{equation}\label{de24}
Q_N(\mu)=0.
\end{equation}
Apparently that one can obtain a simple condition that a solution $\mu$ of equation \reff{de23}
is an eigenvalue of matrix $J$. But we are not a success to
get a sufficient condition that a solution $\mu$ of equation \reff{de24}
is an eigenvalue of matrix $J$, which is more simply than the condition \reff{de21}.
\bigskip

For illustration we consider a few examples. As the first
example we consider the matrix $J$ defined by
\reff{de07}, \reff{de08} for $N=3,\, b_0=b_1=b_2=1$ and for any complex
 $a_0,\, a_1,\, a_2.$

\section{The case $N=3$. The parametric Chebyshev polynomials}\lb{s3}

{\bf 1}. Let us consider the generalized Chebyshev polynomials system
$\basbx{\varphi^{(3)}}{n}{x}$ defined by recurrence relations
\begin{equation}\label{de25}
x\varphi^{(3)}_{n}(x)=\varphi^{(3)}_{n+1}(x)+a_{n}\varphi^{(3)}_{n}(x)+\varphi^{(3)}_{n-1}(x),
\quad\varphi^{(3)}_{0}(x)=1,\, \varphi^{(3)}_{-1}(x)=0.
\end{equation}
The coefficients $a_n$ --- complex numbers that are fulfilled the periodicity condition \reff{de02} with $N=3$.
It is follows from \reff{de25} that the first six polynomials are
\begin{align}
\varphi^{(3)}_{0}&=1& \varphi^{(3)}_{1}&=x-a_0& \varphi^{(3)}_{2}&=(x-a_1)\varphi^{(3)}_{1}-1\nn \\
\varphi^{(3)}_{3}&=(x-a_2)\varphi^{(3)}_{2}-\varphi^{(3)}_{1}& \varphi^{(3)}_{4}&=(x-a_0)\varphi^{(3)}_{3}-
\varphi^{(3)}_{2}& \varphi^{(3)}_{5}&=\varphi^{(3)}_{2}(\varphi^{(3)}_{3}-(x-a_1)),\label{de26}
\end{align}
and for $n\geq6$ they can be calculated by formula
\begin{equation}\label{de27}
\varphi^{(3)}_{n}(x)=P_3(x)\,\varphi^{(3)}_{n-3}(x)-\varphi^{(3)}_{n-6}(x).
\end{equation}
From \reff{de04} and last relation in \reff{de26} it follows that
\begin{equation}\label{de28}
P_3(x)=\varphi^{(3)}_{3}(x)-(x-a_1)=x^3-(a_0+a_1+a_2)x^2+(a_0a_2+a_1a_2+a_0a_1-3)x-a_0a_1a_2+(a_0+a_1+a_2).
\end{equation}
Note that under additional condition $(a_0+a_1+a_2)=0$ the polynomial $P_3(x)$
has more simply form
\begin{equation}\label{28-1}
P_3(x)=x^3+(a_0a_2+a_1a_2+a_0a_1-3)x-a_0a_1a_2.
\end{equation}
Let us find the eigenvalues of the matrix $J^{(3)}=\basc{j^{(3)}}{j,k}$ corresponding to the recurrence relations \reff{de25}
\begin{equation}\label{de29}
j^{(3)}_{i\,k}=B_3\delta_{i+1,k}+A_3\delta_{i,k}+B_3^t\delta_{i-1,k},
\end{equation}
where
\begin{equation}\label{de30}
A_3=\begin{pmatrix}
a_0&1&0\\
1&a_1&1\\
0&1&a_2
\end{pmatrix},\qquad
B_3=\begin{pmatrix}
0&0&0\\
0&0&0\\
1&0&0
\end{pmatrix}.
\end{equation}
 Using the formulas
 \reff{de10-1}-\reff{de10-3}, \reff{de25} and \reff{de26},
one can to write the left-hand side of the equation \reff{de16} in the following form
\begin{equation}\label{de31}
S_0^3(\mu)-D_0^3(\mu)\,P_3(\mu)=\varphi^{(3)}_{2}(\mu)\,Q_3(\mu),
\end{equation}
where
\begin{equation}\label{de32}
Q_3(\mu)=\left(1+\left(\varphi^{(3)}_{1}(\mu)\right)^2\right)(\mu-a_1)(\mu-a_2)-2+\varphi^{(3)}_{2}(\mu)
\left(1-\varphi^{(3)}_{1}(\mu)(\mu-a_2)\right).
\end{equation}
In result the equation \reff{de16} splits into pair of equations
\begin{equation}\label{de33}
\varphi^{(3)}_{2}(\mu)=\mu^2-(a_0+a_1)\mu+a_0a_1-1=0,
\end{equation}
\begin{equation}\label{de34}
Q_3(\mu)=3\mu^2-2(a_0+a_1+a_2)\mu+(a_0a_2+a_1a_2+a_0a_1-3)=0.
\end{equation}
The roots $\mu_{1,2}$ of the equations \reff{de33} have the following form
\begin{equation}\label{de35}
\mu_{1,2}=\mu^{\pm}=\half(a_0+a_1)\pm\half\sqrt{4+(a_1-a_0)^2},
\end{equation}
and the roots of the equations \reff{de34} equal to
\begin{equation}\label{de36}
\mu_{3,4}=\nu^{\pm}=\frac{1}{3}\left((a_0+a_1+a_2)\pm\sqrt{a_0^2+a_1^2+a_2^2-a_0a_1-a_1a_2-a_0a_2+9}\right).
\end{equation}
In the case $(a_0+a_1+a_2)=0$ the roots $\mu_{3,4}$ are simplified
\begin{equation}\label{de37}
\mu_{3,4}=\pm\sqrt{1+\frac{a_0^2+a_1^2+a_2^2}{6}}
\end{equation}
It remains to find those critical values
 $\mu_k,\, (k=1,2,3,4)$ which are  an eigenvalues of Jacobi matrix $J^{(3)}$.
\begin{lemma}
For that a root $\mu_k,\, (k=1,2)$ \reff{de35} of the equation \reff{de33}
is an eigenvalues of the matrix  $J^{(3)}$ \reff{de29}, \reff{de30},
it is necessary and sufficient that the following inequality
\begin{equation}\label{de38}
\left|\mu_k-a_0\right|<1,\quad k=1,2,
\end{equation}
is fulfilled.
\end{lemma}
\begin{proof}[Proof]
In fact, we have for  $\mu_k,\, (k=1,2)$ defined by \reff{de35} the following relation
\begin{equation*}
\sum_{j=0}^{\infty}\left(\varphi^{(3)}_{j}(\mu_k)\right)^2=1+2(\mu_k-a_0)^2\sum_{j=0}^{\infty}(\mu_k-a_0)^{2j}.
\end{equation*}
The series in the right-hand side of this relation is convergent if and only if
the following inequality
\begin{equation*}
 \left|\mu_k-a_0\right|<1,\quad k=1,2,
\end{equation*}
is true.
\end{proof}
Unfortunately, for $\mu_k\ (k=3,4)$ even in simplest case
 $(a_0+a_1+a_2)=0$ one cannot find a more simply condition that
 $\mu_k$ is an eigenvalues of the matrix  $J^{(3)}$ than the following
condition
\begin{equation*}
\left[ \left(\varphi^{(3)}_{3n}(\mu_k)\right)^2+\left(\varphi^{(3)}_{3n+1}(\mu_k)\right)^2+
\left(\varphi^{(3)}_{3n+2}(\mu_k)\right)^2 \right]\rightarrow 0
\end{equation*}
as $n\rightarrow\infty$ (but this is the condition \reff{de21}).
\bigskip

{\bf 2}. We consider the parametric Chebyshev polynomials introduced in \cite{6} as example
to the case when the polynomial
$Q_3(\mu)$ has not roots. These polynomials  $\basbx{\Psi}{n}{x;\alpha}$ are defined by recurrent relations
with coefficients depending on a parameter  $\alpha\in[-1,1]$. These coefficients are
\begin{subequations}\lb{41*}   \begin{eqnarray}
a_0(\alpha)=\frac{i\sqrt{3}}{2}(\alpha+1)(3\alpha-2),\quad a_1(\alpha)=-i\sqrt{3}\alpha,
\quad a_3(\alpha)=-\frac{i\sqrt{3}}{2}(\alpha-1)(3\alpha-2),\\
a_{n+3}(\alpha)=a_n(\alpha),\quad n\geq 0,\qquad  \qquad\qquad\qquad\qquad\qquad
\end{eqnarray} \end{subequations}
It is clear, that  the following equality
\beq\lb{42*}
a_0(\alpha)+a_1(\alpha)+a_2(\alpha)=0,
\eeq
is true.
From \reff{de27}-\reff{de30}, \reff{41*} it is follow that
\begin{align}\lb{43*}
\Psi_0(x;\alpha)&=1;\nn\\
\Psi_1(x;\alpha)&=x-\frac{i\sqrt{3}}{2}(\alpha+1)(3\alpha+2);\nn\\
\Psi_2(x;\alpha)&=x^2-\frac{}{2}(\alpha-1)(3\alpha+2)x+(\frac{3}{2}\alpha(\alpha+1)(3\alpha-2)-1);\nn\\
\Psi_3(x;\alpha)&=x^3+(1-\wt{\alpha}_1^2)x+i\sqrt{3}\alpha(1-\wt{\alpha}_2^2);\\
\Psi_4(x;\alpha)&=(x-a_0(\alpha))\Psi_3(x;\alpha)-\Psi_2(x;\alpha);\nn\\
\Psi_5(x;\alpha)&=\Psi_2(x;\alpha)P_3(x;\alpha),\nn
\end{align}
where are used notations
\begin{subequations}\lb{44*}   \begin{eqnarray}
\wt{\alpha}_1^2=\frac{27}{4}\alpha^2(1-\alpha^2),\quad \wt{\alpha}_2^2=\frac{3}{4}(1-\alpha^2)(9\alpha^2-4)\\
P_3(x;\alpha)=x^3-\wt{\alpha}_1^2x-i\sqrt{3}\alpha\wt{\alpha}_2^2.\qquad\qquad\quad
\end{eqnarray} \end{subequations}
For $n\geq6$ one can calculate the polynomials $\Psi_n(x;\alpha)$ by formulas
\beq\lb{45*}
\Psi_n(x;\alpha)=P_3(x;\alpha)\Psi_{n-3}(x;\alpha)-\Psi_{n-6}(x;\alpha).
\eeq
In the paper \cite{17} we obtain the continuous spectrum of the Jacobi matrix
 $J^{(3)}(\alpha)$, corresponding to the parametric Chebyshev polynomials. The
support of the continuous spectrum is represented on the Fig.1, where are used
the notations
\beq\lb{46*}
\lambda_k=\lambda(\alpha)e^{i\frac{2k\pi}{3}},\qquad
\wt{\lambda}_k=\lambda(\alpha)e^{i\frac{(2k+3)\pi}{3}},\quad k=0,1,2.
\eeq
\begin{figure}[tb]
\centering
\includegraphics[width=6.0cm]{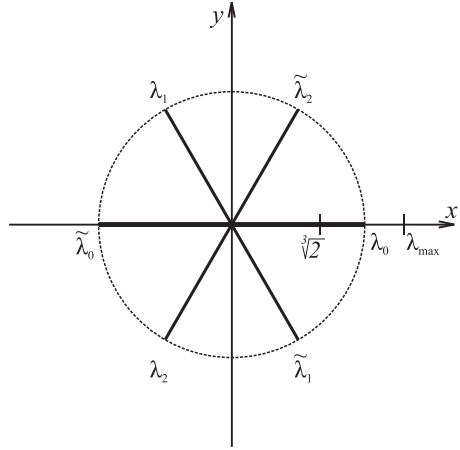}
\caption{The support of the continuous spectrum of the Jacobi matrix  $J^{(3)}(\alpha)$}\label{fig1}
\end{figure}
The number $\lambda(\alpha)\geq0$ is a positive root of the equation
\beq\lb{47*}
\lambda^3-\wt{\alpha}_1^2\lambda-2=0,
\eeq
where
\beq\lb{48*}
\sqrt[3]{2}\leq \lambda(\alpha) \leq \lambda_{max},
\eeq
and
\beq\lb{49*}
\lambda_{max}=\sqrt[3]{1+\sqrt{1-\left(\frac{27}{64}\right)^2}}+\sqrt[3]{1-\sqrt{1-\left(\frac{27}{64}\right)^2}}.
\eeq
Now we consider the discrete spectrum of the matrix $J^{(3)}(\alpha)$. From\reff{de37} and \reff{41*}
we find roots of the equation \reff{de35}
\beq\lb{50*}
\mu_{1,2}(\alpha)=\frac{i\sqrt{3}}{4}(\alpha-1)(3\alpha+2)\pm\sqrt{1-\frac{3}{16}f^2(\alpha)},
\eeq
where
\beq\lb{51*}
f(\alpha)=3\alpha^2+3\alpha-2,
\eeq
and
\beq\lb{52*}
\mu_{1,2}(\alpha)-a_1(\alpha)=-\frac{i\sqrt{3}}{4}f^2(\alpha)\pm\sqrt{1-\frac{3}{16}f^2(\alpha)}.
\eeq
 From lemma 3 it follows that $\mu_k,\, k=1,2$ is an eigenvalue of the matrix $J^{(3)}(\alpha)$ if and
only if the  inequality \reff{de38} is true. We introduce notations
\beq\lb{53*}
\alpha_1=-\half-\sqrt{\frac{1}{4}-\frac{4-2\sqrt{3}}{3\sqrt{3}}},\,
\alpha_2=-\half+\sqrt{\frac{1}{4}-\frac{4-2\sqrt{3}}{3\sqrt{3}}},\,
\alpha_3=-\half+\sqrt{\frac{1}{4}+\frac{4+2\sqrt{3}}{3\sqrt{3}}},
\eeq
and made justification test of the  inequality
 \reff{de38} on intervals
$$
[-1,\alpha_1],\quad (\alpha_1,\alpha_2),\quad [\alpha_2,\alpha_3],\quad (\alpha_3,1].
$$
As a result we have

а) The numbers  $(\mu_{1,2})(\alpha)$ are not eigenvalues of the matrix $J^{(3)}(\alpha)$ as $\alpha\in[-1,\alpha_1)$
or $\alpha\in[\alpha_2,\alpha_3]$.

b) $J^{(3)}(\alpha)$ has only one eigenvalue $\mu_{2}(\alpha)$ as
 $\alpha\in(\alpha_1,\alpha_2)$.

c) $J^{(3)}(\alpha)$ has only one eigenvalue $\mu_{1}(\alpha)$ as
 $\alpha\in(\alpha_3,1]$.

Now we consider the solutions $\mu_{3,4}(\alpha)$ of the equation \reff{de36},
that has in this case the following form
\beq\lb{54*}
Q_3(\mu;\alpha)=3\mu^2-\wt{\alpha}_1^2=0
\eeq
These solutions are equal to
\beq\lb{55*}
\mu_{3,4}(\alpha)=\pm\frac{\wt{\alpha}_1}{\sqrt{3}}.
\eeq
To proof that  $\mu_{3,4}(\alpha)$ are not eigenvalues of the matrix
 $J^{(3)}(\alpha)$ as $\alpha\in[-1,1]$ it is sufficient to check that the
necessary condition \reff{de21} is broken.
Using the recurrence relations \reff{de11} (and notation \reff{de10})
it is easy to show that
$$
S^N_{3n}(\mu)+S^N_{3(n-1)}(\mu)-S(\mu)-S^N_{3}(\mu)=P_3^2(\mu;\alpha)\left(\sigma_{3n,3n+2}(\mu)-S_{2}(\mu)\right).
$$
Then
\beq\lb{56*}
\sum_{s=-1}^1\sigma_{3(n+s),3(n+s)+2}(\mu)+(1-P_3^2(\mu;\alpha))\sigma_{3n,3n+2}(\mu)=
S_2+S_1(2-P_3^2(\mu;\alpha))+\sigma_{6,6+2}(\mu)=A_0(\alpha).
\eeq
The direct calculation shows that for any
$\alpha\in(-1,0)\cup(0,1]$
one have $A_0(\alpha)\neq0$.
At the same time the left-hand side of equality \reff{56*} tends to zero
as $n\rightarrow\infty$ if the nesessary condition \reff{de21} is fulfilled.
From this it follows that $A_0(\alpha)=0$. Thus $\mu_{3,4}(\alpha)$ are not
eigenvalues of the Jacobi matrix $J^{(3)}(\alpha)$ as $\alpha\in(-1,0)\cup(0,1]$.
 From results of \cite{18} it follows that $\mu_{3,4}(0)$ and $\mu_{3,4}(1)$
are not eigenvalues of the Jacobi matrixes $J^{(3)}(0)$ and $J^{(3)}(1)$ too.
Therefore, $\mu_{3,4}(\alpha)$ are not eigenvalues of the Jacobi matrix
$J^{(3)}(\alpha)$ as $\alpha\in[-1,1]$.
\bigskip

So the Jacobi matrix   $J^{(3)}(\alpha)$ has only one eigenvalue $\mu_{1}(\alpha)$
for $\alpha\in(\alpha_3,1],$ and only one eigenvalue $\mu_{2}(\alpha)$ for $\alpha\in(\alpha_1,\alpha_2)$.
\bigskip

Further, as another example, we consider the Jacobi matrix for the elementary
$N$-symmetric Chebyshev polynomials, which belongs to the type of polynomials
under consideration. These polynomials appear in studying of "compound model of
generalized oscillator" \cite{14}-\cite{16}. We consider only cases
 $N=3,4,5,$ since it was shown in the paper \cite{14} that such polynomials not
exist for $n\geq6$.

\section{Elementary $N$-symmetric Chebyshev polynomials}\lb{s4}

Elementary $N$-symmetric Chebyshev polynomials $\basbx{\varphi^N}{n}{x}$ \cite{14} are
defined by recurrence relations
\begin{equation}\label{de39}
x\varphi^N_{n}(x)=\varphi^N_{n+1}(x)+a_{n}\varphi^N_{n}(x)+\varphi^N_{n-1}(x),
\quad\varphi^N_{0}(x)=1,\, \varphi^N_{-1}(x)=0,
\end{equation}
where coefficients $a_n$ for $N=3,4,5,$ given by formulas
\begin{subequations} \lb{de40}
\begin{eqnarray}
a_0^{(3)}=i\sqrt{3},\, a_1^{(3)}=i\sqrt{3},\, a_2^{(3)}=0,\,\, a_{n+3}^{(3)}=a_n^{(3)},\, n\geq0;\\
a_0^{(4)}=2i,\, a_1^{(4)}=0,\, a_2^{(4)}=-2i,\, a_3^{(4)}=0,\,\, a_{n+4}^{(4)}=a_n^{(4)},\, n\geq0;\\
a_0^{(5)}=a_2^{(5)}=a_3^{(5)}=0,\, a_1^{(5)}=i\sqrt{5},\, a_4^{(5)}=-i\sqrt{5},\,\, a_{n+5}^{(5)}=a_n^{(5)},\, n\geq0.
\end{eqnarray} \end{subequations}
Using recurrence relations \reff{de39}, we find first $2N$ ($N=3,4,5$)polynomials
\begin{align}\lb{de41-1}
\varphi_{0}^{(3)}(x)&=1,& \varphi_{1}^{(3)}(x)&=x-i\sqrt{3},& \varphi_{2}^{(3)}(x)&=x^2+2,\nn \\
\varphi_{3}^{(3)}(x)&=x\varphi_{2}^{(3)}(x)-\varphi_{1}^{(3)}(x),& \varphi_{4}^{(3)}(x)&=
x^3\varphi_{1}^{(3)}(x)+1,& \varphi_{5}^{(3)}(x)&=x^3\varphi_{2}^{(3)}(x);
\end{align}
\begin{align}\label{de41-2}
\varphi_{0}^{(4)}(x)&=1,& \varphi_{1}^{(4)}(x)&=x-2i,\nn\\
\varphi_{2}^{(4)}(x)&=x^2-2ix-1,& \varphi_{3}^{(4)}(x)&=x^3+2x,\nn\\
\varphi_{4}^{(4)}(x)&=x^4+x^2+2ix+1,& \varphi_{5}^{(4)}(x)&=x^5-2ix^4+3x-2i,\\
\varphi_{6}^{(4)}(x)&=x^6-2ix^5-x^4+2x^2-4ix-1,& \varphi_{7}^{(4)}(x)&=(x^4+2)\varphi_{3}^{(4)}(x);\nn
\end{align}
\begin{align}\label{de41-3}
\varphi_{0}^{(5)}(x)&=1,& \varphi_{1}^{(5)}(x)&=x,\nn\\
\varphi_{2}^{(5)}(x)&=x^2-i\sqrt{5}x-1,& \varphi_{3}^{(5)}(x)&=x^3-i\sqrt{5}x^2-2x,\nn\\
\varphi_{4}^{(5)}(x)&=x^4-i\sqrt{5}x^3-3x^2+i\sqrt{5}x+1,& \varphi_{5}^{(5)}(x)&=x^5+x^3-i\sqrt{5}x^2-2x+i\sqrt{5},\\
\varphi_{6}^{(5)}(x)&=x^6+x^2-1,& \varphi_{7}^{(5)}(x)&=x^7-i\sqrt{5}x^6-x^5+x,\nn\\
\varphi_{8}^{(5)}(x)&=x^8-i\sqrt{5}x^7-2x^6+1,& \varphi_{9}^{(5)}(x)&=x^5\varphi_{4}^{(5)}(x).\nn
\end{align}
From here and relation \reff{de04} we get the following expression for $P_N(x)$
\begin{equation}\label{de42}
P_3(x)=x^3,\quad P_4(x)=x^4+2,\quad P_5(x)=x^5.
\end{equation}
In view of \reff{de42}, the relation \reff{de06} takes the following form
 $(k=0,1,2,\, m\geq2)$
\begin{subequations}\label{de43}
\begin{eqnarray}
\varphi_{3m+k}^{(3)}(x)=\varphi_{k+3}^{(3)}(x)U_{m-1}(x^3)-\varphi_k^{(3)}(x)U_{m-2}(x^3);\label{de43-1}\\
\varphi_{4m+k}^{(4)}(x)=\varphi_{k+4}^{(4)}(x)U_{m-1}(x^4+2)-\varphi_k^{(4)}(x)U_{m-2}(x^4+2);\label{de43-2}\\
\varphi_{5m+k}^{(5)}(x)=\varphi_{k+5}^{(5)}(x)U_{m-1}(x^5)-\varphi_k^{(5)}(x)U_{m-2}(x^5). \label{4de3-3}
\end{eqnarray}
\end{subequations}

The Jacobi matrix $J^{(N)}=\basc{j^{(N)}}{j,k},$ $(N=3,4,5)$ corresponding to the
relations  \reff{de39}, has the form
\begin{equation}\label{de44}
J^{(N)}=B_N\delta_{i+1,k}+A_N\delta_{i,k}+B_N^t\delta_{i-1,k},
\end{equation}
where
\begin{subequations}\label{de45}
 \begin{eqnarray}
 A_3=\begin{pmatrix}
 i\sqrt{3}&1&0\\
 1&-i\sqrt{3}&1\\
 0&1&0
 \end{pmatrix}
 ,\qquad\qquad\qquad
 B_3=\begin{pmatrix}
 0&0&0\\
 0&0&0\\
 1&0&0
 \end{pmatrix};
\label{de45-1} \\[8pt]
 A_4=\begin{pmatrix}
 2i&1&0&0\\
 1&0&1&0\\
 0&1&-2i&0\\
 0&0&1&0
 \end{pmatrix}
 ,\qquad\qquad
 B_4=\begin{pmatrix}
 0&0&0&0\\
 0&0&0&0\\
 0&0&0&0\\
 1&0&0&0
 \end{pmatrix};
\label{de45-2} \\[8pt]
A_4=\begin{pmatrix}
 0&1&0&0&0\\
 1&i\sqrt{5}&1&0&0\\
 0&1&0&1&0\\
 0&0&1&0&1\\
 0&0&0&1&-i\sqrt{5}
 \end{pmatrix}
 ,\qquad
 B_4=\begin{pmatrix}
 0&0&0&0&0\\
 0&0&0&0&0\\
 0&0&0&0&0\\
 0&0&0&0&0\\
 1&0&0&0&0
 \end{pmatrix};
\label{de45-3}
\end{eqnarray}
\end{subequations}

Now we turn to evaluation of eigenvalues of the matrix
 $J^{(N)},\,\,(N=3,4,5),$ using the critical equation \reff{de16}.
\bigskip

\centerline{\underline{A. Discrete spectrum of Jacobi matrix $J^{(3)}$}}
\bigskip

The Jacobi matrix $J^{(3)}$ is defined by equalities
 \reff{de44} and \reff{de45-1}. From \reff{de31}, \reff{de41-1}, \reff{de42}, \reff{de34}, \reff{de10-1}, taking into account that
\begin{equation*}
a_0^{(3)}+a_1^{(3)}+a_2^{(3)}=0,\qquad a_0^{(3)}a_1^{(3)}+a_1^{(3)}a_2^{(3)}+a_0^{(3)}a_2^{(3)}=3,
\end{equation*}
we obtain that left-hand side
 $\Delta_0^{(3)}(\mu)$ of equation \reff{de16} has the following form
\begin{equation*}
\Delta_0^{(3)}(\mu)=S_0^{(3)}(\mu)-D_0^{(3)}(\mu)\,P_3(\mu)=3\mu^2(\mu^2+2).
\end{equation*}
Then the equation \reff{de16}  for  $N=3$ looks as
\begin{equation*}
\mu^2(\mu^2+2)=0.
\end{equation*}
The solutions of this equation are equal to
\begin{equation}\label{de46}
\mu_1=i\sqrt{2},\quad \mu_2=-i\sqrt{2},\quad \mu_{3,4}=0.
\end{equation}
Using the lemma 4, we have
\begin{equation*}
|\mu_1-a_0|=|i\sqrt{2}-i\sqrt{3}|<1.
\end{equation*}
It means that $\mu_1$ is an eigenvalue of the matrix $J^{(3)}.$
Further,
\begin{equation*}
|\mu_2-a_0|=|-i\sqrt{2}-i\sqrt{3}|>1,
\end{equation*}
i.e. $\mu_2$ is not an eigenvalue of the matrix $J^{(3)}.$
Now we calculate the components of vector $X_{\mu_{3}}=X_{\mu_{4}}$
for $\mu_{3,4}$. They are equal to
\begin{equation}\label{de47}
x_1=1,\, x_2=-i\sqrt{3},\, x_3=2,\, x_4=i\sqrt{3},\, x_5=1,\, x_6=0,\quad x_{k+6}=-x_k,\quad\text{при}\quad k\geq0.
\end{equation}
Taking into account that  $\varphi_{k-1}^{(3)}(\mu)=x_k,$ for $k=3,4$, we have
\begin{equation*}
\left[
\left( \varphi_{3n}^{(3)}(\mu_k)\right)^2+\left( \varphi_{3n+1}^{(3)}(\mu_k)\right)^2+
\left( \varphi_{3n+2}^{(3)}(\mu_k)\right)^2
\right]\nrightarrow 0, \quad\text{при}\quad n\rightarrow\infty,
\end{equation*}
i.e. the condition \reff{de21} is not realized. Therefore $X_{\mu_{k}}\notin \ell^2$, it means that
$\mu_{3,4}$ are not  eigenvalues of the matrix $J^{(3)}$.
\bigskip

\centerline{\underline{B. Discrete spectrum of Jacobi matrix $J^{(4)}$}}

\bigskip

The Jacobi matrix $J^{(4)}$ is defined by equalities \reff{de44} and \reff{de45-2}.
Using \reff{de31}, \reff{de41-2}, \reff{de42}, \reff{de34}, \reff{de10-2}, we rewrite the equation  \reff{de16} for $N=4$
in the form
\begin{equation}\label{de48}
\mu^4(\mu^2+2)=0.
\end{equation}
The solutions of this equation are equal to
\begin{equation*}
\mu_1=i\sqrt{2},\quad \mu_2=-i\sqrt{2},\quad \mu_{3,4,5,6}=0.
\end{equation*}
In the first place we consider zero solutions  of this equation.
The components of the vector $X_{\mu_{k}}$ (for $k=3,4,5,6$) are equal to
\begin{equation*}
x_1=1,\, x_2=-\frac{i}{2},\, x_3=-1,\, x_4=-\frac{3i}{2},\, \quad x_{k+4}=x_k,\quad\text{при}\quad k\geq0.
\end{equation*}
Then partial sum of the series
\begin{equation*}
\sum_{n=0}^{\infty}\left( \varphi_{n}^{(4)}(\mu_k)\right)^2=\sum_{n=1}^{\infty}x_k^{\,\,2}
\end{equation*}
equal to
\begin{equation*}
S_1=1,\, s_2=1-\frac{i}{2},\, S_3=-\frac{i}{2},\, S_4=-2i,\quad S_{n+4}=S_n-2i, \quad\text{при}\quad n\geq1.
\end{equation*}
Thus the sequence of partial sums of series
$\sum_{n=1}^{\infty}x_k^{\,\,2}$ has not limit point as
$n\rightarrow\infty$ , i.e. the series is divirgent.
Hence, $\mu_{3,4,5,6}$ are not  eigenvalues of the matrix
$J^{(4)}.$
The components of vector $X_{\mu_{1}}=(x_1,x_2,\ldots)^t$, corresponding to
critical value  $\mu_1$, are equal to
\begin{equation*}
x_1=1,\, x_2=i(\sqrt{2}-2),\, x_3=(2\sqrt{2}-3),\, x_4=0,\, \quad x_{k+4}=(3-2\sqrt{2})x_k,\quad\text{при}\quad k\geq1.
\end{equation*}
Because
\begin{equation*}
    \|X_{\mu_{1}}\|^2=\sqrt{2},
\end{equation*}
then critical value $\mu_1=i\sqrt{2}$ is eigenvalue of the Jacobi matrix $J^{(4)}.$
The normalized eigenvector $Y_{\mu_{1}}=(y_1,y_2,\ldots)^t$ has the
following components
\begin{equation*}
y_{4k+1}=\frac{1}{\sqrt[4]{2}}\,(3-2\sqrt{2})^k,\,\,
y_{4k+2}=\frac{i(\sqrt{2}-2)}{\sqrt[4]{2}}\,(3-2\sqrt{2})^k,\,\,
y_{4k+3}=\frac{(2\sqrt{2}-3)}{\sqrt[4]{2}}\,(3-2\sqrt{2})^k,\,\,
y_{4k+4}=0,
\end{equation*}
with $k\geq0$.
Finally, the components of vector $X_{\mu_{2}}=(x_1,x_2,\ldots)^t$,
corresponding to  critical value $\mu_2=-i\sqrt{2}$, are equal to
\begin{equation*}
x_1=1,\, x_2=-i(2+\sqrt{2}),\, x_3=-(3+2\sqrt{2}),\, x_4=0,\, \quad x_{k+4}=(3+2\sqrt{2})x_k,\quad\text{при}\quad k\geq1.
\end{equation*}
Consequently,
\begin{equation*}
    \|X_{\mu_{2}}\|^2=\infty
\end{equation*}
i.e. critical value $\mu_2=-i\sqrt{2}$ is not eigenvalue of the Jacobi matrix
 $J^{(4)}.$
\bigskip

\centerline{\underline{C. Discrete spectrum of Jacobi matrix $J^{(5)}$}}

\bigskip

We consider now the Jacobi matrix  $J^{(5)}$ that is defined by equalities
 \reff{de44} and \reff{de45-3}. Using  \reff{de31}, \reff{de41-3}, \reff{de42}, \reff{de34}, \reff{de10-3},
we rewrite the equation \reff{de16} for $N=5$ in the form
\begin{equation}\label{de49}
\mu^4(\mu^4-i\sqrt{5}\mu^3-3\mu^2+i\sqrt{5}\mu+1)=0.
\end{equation}
The solutions of this equation are equal to
\begin{gather}\label{de50}
\mu_{1,2}=\frac{1}{4}\left[ \pm\sqrt{10-2\sqrt{5}}+i(1+\sqrt{5}) \right],\nn\\
\mu_{3,4}=\frac{1}{4}\left[ \pm\sqrt{10+2\sqrt{5}}+i(-1+\sqrt{5}) \right], \\
\mu_{5,\,6,\,7,\,8}=0.\nn
\end{gather}
From the same arguments as given above, we see that for
$k=\overline{5;8}$ the vector $X_{\mu_{k}}\notin \ell^2,$
i.e. corresponding critical values $\mu_{k}$ are not eigenvalues of the Jacobi matrix $J^{(5)}.$

For the critical value $\mu_{1}=\frac{1}{4}\left[ \sqrt{10-2\sqrt{5}}+i(1+\sqrt{5}) \right]$
the squared components of vector $X_{\mu_{1}}=(x_1,x_2,\ldots)^t$ are equal to
\begin{gather*}
(x_1)^2=1,\quad (x_2)^2=\frac{1-\sqrt{5}}{4}+\frac{i(1+\sqrt{5})}{8}\sqrt{10-2\sqrt{5}},\\[5pt]
(x_3)^2=\frac{\sqrt{5}-2}{2}+\frac{i(1-\sqrt{5})}{8}\sqrt{10-2\sqrt{5}},\quad (x_4)^2=\frac{\sqrt{5}-3}{2},\\[5pt]
(x_5)^2=0,\qquad (x_{k+5})^2=\frac{\sqrt{5}-3}{2}\,(x_{k})^2,\quad k\geq1,
\end{gather*}
Then $ \|X_{\mu_{1}}\|^2=2\sqrt{5}$ and consequently, $\mu_{1}$
is an eigenvalue of the Jacobi matrix   $J^{(5)}.$

Similarly, for critical value $\mu_{2}=\frac{1}{4}\left[ -\sqrt{10-2\sqrt{5}}+i(1+\sqrt{5}) \right]$ the squared
components of vector $X_{\mu_{2}}$  are equal to
\begin{gather*}
(x_1)^2=1,\quad (x_2)^2=\frac{1-\sqrt{5}}{4}-\frac{i(1+\sqrt{5})}{8}\sqrt{10-2\sqrt{5}},\\[5pt]
(x_3)^2=\frac{\sqrt{5}-2}{2}+\frac{i(\sqrt{5}-1)}{8}\sqrt{10-2\sqrt{5}},\quad (x_4)^2=\frac{\sqrt{5}-3}{2},\\[5pt]
(x_5)^2=0,\qquad (x_{k+5})^2=\frac{\sqrt{5}-3}{2}\,(x_{k})^2,\quad k\geq1.
\end{gather*}
Therefore, $ \|X_{\mu_{2}}\|^2=2\sqrt{5}$ and
 $\mu_{2}$ is an eigenvalue of the Jacobi matrix   $J^{(5)}.$

For critical value $\mu_{3}=\frac{1}{4}\left[ \sqrt{10+2\sqrt{5}}+i(\sqrt{5}-1) \right]$ the  squared
components of vector $X_{\mu_{3}}$  are equal to
\begin{gather*}
(x_1)^2=1,\quad (x_2)^2=\frac{1+\sqrt{5}}{4}+\frac{i(\sqrt{5}-1)}{8}\sqrt{10+2\sqrt{5}},\\[5pt]
(x_3)^2=-\frac{2+\sqrt{5}}{2}-\frac{i(1+\sqrt{5})}{8}\sqrt{10+2\sqrt{5}},\quad (x_4)^2=-\frac{\sqrt{5}+3}{2},\\[5pt]
(x_5)^2=0,\qquad (x_{k+5})^2=-\frac{3+\sqrt{5}}{2}\,(x_{k})^2,\quad k\geq1.
\end{gather*}
From these relations follow that $X_{\mu_{3}}\notin \ell^2,$ and
corresponding critical value $\mu_{3}$ is not eigenvalue of the Jacobi matrix $J^{(5)}.$

Finally, for critical value $\mu_{4}=-\frac{1}{4}\left[ \sqrt{10+2\sqrt{5}}-i(\sqrt{5}-1) \right]$
the squared components of vector $X_{\mu_{4}}$  are equal to
\begin{gather*}
(x_1)^2=1,\quad (x_2)^2=\frac{1+\sqrt{5}}{4}-\frac{i(\sqrt{5}-1)}{8}\sqrt{10+2\sqrt{5}},\\[5pt]
(x_3)^2=-\frac{2+\sqrt{5}}{2}+\frac{i(1+\sqrt{5})}{8}\sqrt{10+2\sqrt{5}},\quad (x_4)^2=-\frac{\sqrt{5}+3}{2},\\[5pt]
(x_5)^2=0,\qquad (x_{k+5})^2=-\frac{3+\sqrt{5}}{2}\,(x_{k})^2,\quad k\geq1.
\end{gather*}
As above, from these relations it follows that   $X_{\mu_{4}}\notin \ell^2$ and
corresponding critical value $\mu_{4}$ is not eigenvalue of the Jacobi matrix $J^{(5)}.$


\begin{thebibliography}{99}
\bibitem{1} M. Kac, P. Van Moerbeke, {\it On some periodic Toda lattices}. Proc. Nat. Acad. Sci. U.S.A. (2nd ed.),
{\bf 72}, 1627–1629  (1975).
\bibitem{2} Аптекарев А. И. {\it Асимптотические свойства многочленов,
ортогональных на системе контуров, и периодические движения цепочек Тода}.
Мат. Сб. {\bf 125(167)}, № 2(10),  231–258 (1984). ( Math. USSR-Sb., 53, pp. 233–260 (1986)).
\bibitem{3} W. Van Assche, {\it Christoffel functions and Tura'n determinants on several intervals}.
   J. Comput. and Appl. Math. {\bf 48}:1–2, 207–223 (1993).
\bibitem{4} Д. Барриос, Г. Лопес, Э. Торрано,
{\it Полиномы, порожденные трехчленным рекуррентным соотношением с
асимптотически периодическими комплексными коэффициентами}.
Мат. Сб. {\bf 186}:5, 3–34 (1995) (Mat. Sb., {\bf 186}:5, 3–34 (1995)).
\bibitem{5} J. Bazargan, I. Egorova, {\it Jacobi operator with step-like asymptotically periodic coefficients}.
   Mat. Fiz. Anal. Geom., {\bf 10}:3, 425–442 (2003).
\bibitem{6} J. Geronimo, W. Van Assche, {\it Orthogonal polynomials with asymptotically periodic recurrence сoefficients}.
J. Approx. Theory {\bf 46}, 251–283 (1986).
\bibitem{6a} J. Gilewicz, E. Leopold, {\it Zeros of polynomials and recurrence relations with
periodic coefficients}. J. Comput. Appl. Math. {\bf 107}:2, 241–255 (1999)
\bibitem{7} C.C. Grosjean, {\it The measure induced by orthogonal polynomials satisfying a recursion formula
with either constant or periodic coefficients}. Part I: {\it Constant coefficients}.
Med. Konink. Acad. Wetensch. Belgie, {\bf 48}:3, 39–60 (1986).
\bibitem{8} P Van Moerbeke, {\it The spectrum of Jacobi matrices}. Invent. Math. {\bf 37}:1, 45-81 (1976). 
\bibitem{9} F. Peherstorfer, {\it On Bernstein-Szego orthogonal polynomials on several intervals}. II.
{\it Orthogonal polynomials with periodic recurrence coefficients}. J. Approx. Theory {\bf 64}:2, 123–161 (1991).
\bibitem{10} F. Peherstorfer, R. Steinbauer, {\it  Orthogonal polynomials on arcs of the unit circle}. II.
 {\it Orthogonal polynomials with periodic reflection coefficients}. J. Approx. Theory {\bf 87}:1, 60–102 (1996).
\bibitem{11} F. Peherstorfer, R. Steinbauer, {\it Asymptotic Behaviour of Orthogonal Polynomials on the Unit Circle
with Asymptotically Periodic Reflection Coefficients}. J.  Approx. Theory {\bf 88}:3, 316–353 (1997).
\bibitem{12} A. Almendral Va'zquez, {\it The Spectrum of a Periodic Complex Jacobi Matrix Revisited}.
 J. Approx. Theory {\bf 105}:2, 344–351 (2000).
\bibitem{13} B. Beckermann, J. Gilewicz, E. Leopold, {\it Recurrence relation with periodic coefficients and
Chebyshev polynomials},
Applicationes Mathematicae {\bf 23}, 319-323  (1995).
\bibitem{14}  В.В. Борзов,  Е.В. Дамаскинский, {\it $N$-симметричные полиномы Чебышева в составной модели обобщенного
осциллятора}, ТМФ {\bf 129}, №2, 229-240 (2011).
\bibitem{15} В.В. Борзов,  Е.В. Дамаскинский, {\it  Составная модель обобщенного осциллятора. I },  ЗНС ПОМИ
{\bf 374}, 58-81 (2010).
\bibitem{16} V.V. Borzov, E.V. Damaskinsky, {\it Connection between representations of nonstandard and standard
Chebyshev oscillators}, Day on Diffraction 2010 28-34.
\bibitem{17} V.V. Borzov, E.V. Damaskinsky, {\it The differential equation for generalized parametric Chebyshev polynomials},
Days on Diffraction 2012.
\bibitem{18}   В.В. Борзов,  Е.В. Дамаскинский, {\it Дифференциальные уравнения для простейших 3-сим\-метричных полиномов
Чебышева}. Записки Научных Семинаров ПОМИ, т.398,64-86 (2012).
\end {thebibliography}
\end{document}